\pdfoutput=1
\documentclass[11pt,letterpaper]{amsart}

\usepackage[utf8]{inputenc}
\usepackage[margin=1in]{geometry}
\usepackage[T1]{fontenc}

\usepackage{lmodern}
\usepackage{mathtools,amsfonts,amsthm,mathrsfs,amssymb}
\usepackage{thmtools,thm-restate}

\usepackage{textcomp}
\mathtoolsset{showonlyrefs}

\usepackage[shortlabels]{enumitem}

\PassOptionsToPackage{pdfusetitle}{hyperref}
\usepackage{bookmark}
\usepackage{hypcap}

\newtheorem{theorem}{Theorem}
\newtheorem{lemma}[theorem]{Lemma}
\newtheorem{claim}[theorem]{Claim}

\newtheorem{conjecture}[theorem]{Conjecture}

\theoremstyle{definition}

\newcommand{\oldqed}{}
\def\endofFact{\hfill\scalebox{.6}{$\blacksquare$}}
\newenvironment{claimproof}[1][Proof]{
  \renewcommand{\oldqed}{\qedsymbol}
  \renewcommand{\qedsymbol}{\endofFact}
  \begin{proof}[#1]
}{
  \end{proof}
  \renewcommand{\qedsymbol}{\oldqed}
}

\newcommand*{\lam}{\lambda}

\renewcommand{\epsilon}{\varepsilon}
\newcommand{\eps}{\varepsilon}

\renewcommand{\subset}{\subseteq}

\let\oldtriangle\triangle
\def\triangle{\mathbin{\oldtriangle}}
\newcommand*{\Ccut}{\mathcal{C}^{\mathrm{cut}}}
\newcommand*{\Cind}{\mathcal{C}^{\mathrm{ind}}}

\DeclareMathOperator{\Tr}{Tr}
\newcommand*{\transpose}{\mathrm{T}}

\newlist{enumerata}{enumerate}{1}
\setlist[enumerata]{label=\upshape{(\alph*)}}
\setlist[enumerate]{label=\upshape{(\roman*)}}

\title[Approximately counting independent sets in dense bipartite graphs]{Approximately counting independent sets in dense bipartite graphs via subspace enumeration}
\date{\today}

\author[C.\ Carlson]{Charlie Carlson}
\address{Department of Computer Science, University of Colorado Boulder, USA}
\email{chca0914@colorado.edu}

\author[E.\ Davies]{Ewan Davies}
\address{Department of Computer Science, Colorado State University, Fort Collins, USA}
\email{research@ewandavies.org}

\author[A.\ Kolla]{Alexandra Kolla}
\address{Department of Computer Science and Engineering, University of California Santa Cruz, USA}
\email{akolla@ucsc.edu}

\author[A.\ Potukuchi]{Aditya Potukuchi}
\address{Department of Electrical Engineering and Computer Science, York University, Toronto, Canada}
\email{apotu@yorku.ca}

\begin{document}
\begin{abstract}
We give a randomized algorithm that approximates the number of independent sets in a dense, regular bipartite graph---in the language of approximate counting, we give an FPRAS for \#BIS on the class of dense, regular bipartite graphs. 
Efficient counting algorithms typically apply to ``high-temperature'' problems on bounded-degree graphs, and our contribution is a notable exception as it applies to dense graphs in a low-temperature setting. 
Our methods give a counting-focused complement to the long line of work in combinatorial optimization showing that CSPs such as Max-Cut and Unique Games are easy on dense graphs via spectral arguments. The proof exploits the fact that dense, regular graphs exhibit a kind of small-set expansion (i.e.\ bounded threshold rank), which via subspace enumeration lets us enumerate small cuts efficiently.
\end{abstract}

\maketitle

\section{Introduction}

Exactly computing the number $i(G)$ of independent sets in a graph $G$ is \#P-hard, even when restricted to bipartite graphs~\cite{PB83}. In the general case, approximating $i(G)$ (to within, say, a constant factor) is NP-hard, even when restricted to $d$-regular graphs with $d\ge 6$~\cite{GGS+11,Sly10,SS12}.
Restricted to \emph{bipartite} graphs the problem of counting independent sets is known as \#BIS, and the prospect of hardness of approximation is less clear because finding a maximum independent set can be done in polynomial time.
Under polynomial-time approximation-preserving reductions, many natural counting problems are equivalent to \#BIS~\cite{DGGJ04}, and the complexity of approximating \#BIS has received a lot of attention. 
Existing approximation algorithms for \#BIS include ``high-temperature'' algorithms that work when degrees on one side of the bipartition are small~\cite{LL15}, ``low-temperature'' algorithms that require additional assumptions such as expansion~\cite{CGG+19,JKP19b} or unbalanced degrees~\cite{CPSODA}, and exponential-time algorithms that are nonetheless faster than algorithms for the general, non-bipartite case~\cite{GLR21}.
The description of these methods in terms of temperature is due to a common generalization in terms of weighted counting and strong connections to statistical physics, where counting (weighted) independent sets corresponds to computing the partition function of the hard-core model.

The idea that Max-CSP optimization problems such as Max-Cut and Unique Games should be easy to approximate on dense graphs---perhaps because they have good expansion properties---is well-established~\cite{AKK95,Fri96,FK96}.
Many of the techniques that apply to dense or expanding graphs have been generalized in interesting directions. 
In particular, spectral methods give good results in both dense graphs and expanders, and in many cases can be extended to more refined structural properties such as small-set expansion and threshold rank to great effect.
Most of the prominent approaches to Max-CSPs relevant to this work fall into three categories: algorithmic regularity lemmas which began with Frieze and Kannan~\cite{FK96} and were extended to threshold rank by Oveis Gharan and Trevisan~\cite{OT13a}; convex hierarchies and correlation rounding~\cite{AKK+08,BRS11a,GS11a}; and the spectral technique of subspace enumeration due to Kolla and Tulsiani~\cite{Kol10,KT07}. 
Prior to these developments were several algorithms demonstrating that counting problems on dense graphs admit efficient approximation algorithms~\cite{Ann94,DFJ94,JS89}, though these results do not apply to counting independent sets. 

An analogous theme in approximate counting is to obtain algorithms on expander graphs or random graphs~\cite{BGG+18,CDF+22a,GGS20,HJP23,JKP19b}.
Despite superficial similarity to the aforementioned work on Max-CSPs in the sense that these works give algorithms for dense or expanding instances, there is relatively little work establishing any common underlying phenomenon that makes Max-CSP problems and counting problems easy on dense or expanding graphs. 
A notable exception is due to Risteski~\cite{Ris16}, who connected the work on correlation rounding and convex hierarchies~\cite{BRS11a} to the broad and well-studied problem of approximating partition functions. His approach is also known as the variational method.
Regularity methods and correlation rounding do provide some evidence of structure common to these problems; for example, Coja--Oghlan and various coauthors have developed a range of regularity lemmas and applied them to both Max-CSPs and spin models on random graphs~\cite{BC16,CCF10,CP16}, and independently discovered correlation rounding in the context of Gibbs measures and partition functions~\cite{CP19a}. 
Counting independent sets is not typically one of the examples studied, though occasionally this is more for convenience than for fundamental reasons. 

In the specific context of \#BIS, connections to Max-CSP research are even more scarce.
The polymer approach of Jenssen, Keevash and Perkins~\cite{JKP19b} is a major algorithmic breakthrough for \#BIS which shows that several prominent \#BIS-hard problems can be approximated in polynomial time on bounded-degree expander graphs (and thus random $d$-regular graphs for $d=O(1)$). 
Further refinements of the method broaden the range of problems covered~\cite{GGS20,HJP23}, provide faster algorithms based on rapid mixing of Markov chains known as polymer dynamics~\cite{CGG+19}, or weaken the structural properties required by applying container theorems to combinatorial enumeration problems that arise in the method~\cite{CDF+22a,JPP22}. 
None of these developments give polynomial-time algorithms in dense graphs, however,
Carlson, Davies, and Kolla~\cite{CDK20} applied the polymer method to approximate the Potts model partition function on (bounded-degree) graphs with bounded threshold rank, but the conditions their analysis requires are prohibitively restrictive, and it is unclear whether their techniques can be applied to \#BIS.
While Risteski's approach has been extended and improved~\cite{JKR19,KLR22a}, results are stated for spin models with soft constraints such as the Ising and Potts models, and the approximation guarantees degrade in the presence of the hard constraints that are inherent to independent sets.

\subsection{Main result}

We specifically address the superficial similarities between algorithms for Max-CSPs and counting independent sets by giving an algorithm for approximately counting independent sets in dense, regular bipartite graphs which combines the highly successful techniques of polymer models, subspace enumeration, and container theorems for the enumeration of independent sets in bipartite graphs. Our approximation guarantee is of the strong type typically sought in approximate counting. We say that a relative $\eps$-approximation of a real number $x$ is a real number $y$ such that $e^{-\eps} \le x/y \le e^{\eps}$, and a fully polynomial randomized approximation scheme (FPRAS) is an algorithm that with probability at least $3/4$ outputs a relative $\eps$-approximation in time polynomial in the instance size and $1/\eps$.

\begin{theorem}\label{thm:main}
    For each $\delta\in(0,1)$ there is an FPRAS for \#BIS on the class of $\lfloor \delta n\rfloor$-regular bipartite graphs.
\end{theorem}

We use spectral methods and subspace enumeration to enumerate small cuts in $d$-regular bipartite graphs via an $\eps$-net of the vector space spanned by small eigenvalues of the Laplacian matrix of the graph, influenced by the use of these methods in combinatorial optimization~\cite{ABS10,Kol10,KT07} and approximate counting. 
Some of our analysis builds upon the perturbative approach of~\cite{HPR19a,JKP19b} and
an important refinement of this method due to Jenssen and Perkins~\cite{JP20} (and with Potukuchi~\cite{JPP21}) that uses graph container lemmas of the type developed by Sapozhenko~\cite{S87,S01b}. 
While container theorems for independent sets have been used to control enumeration problems that arise in establishing the convergence of the cluster expansion~\cite{JP20,JPP21,JPP22}, and these have inspired container-like theorems for controlling analogous enumeration problems~\cite{CDF+22a}, our addition of subspace enumeration here has a different purpose. 

In terms of running time, our result improves upon the dense case of an algorithm of Jenssen, Perkins, and Potukuchi~\cite{JPP22} which runs in subexponential time on $d$-regular bipartite graphs for all $d\ge\omega(1)$. In the case $d=\Theta(n)$ their algorithm takes time $\exp(\Omega(\log^4 n))$, and our contribution works for any accuracy parameter $\eps$, which is not given by the methods in~\cite{JPP22}. 
The improvement stems from incorporating the spectral techniques mentioned above, which lets us sidestep algorithmic cluster expansion. 
That is, our spectral techniques overcome an obstacle in the algorithm of~\cite{JPP22} related to polynomial accuracy: we can achieve arbitrary accuracy without resorting to a naive enumeration of polymers (which in this setting are connected subgraphs of the square of the instance).

An interesting question posed in~\cite{JPP22} is whether \#BIS admits a general subexponential-time algorithm. 
One of our technical contributions is to show that a perspective on graph spectra involving higher-order eigenvalues and eigenvectors advances our understanding of \#BIS.

\section{Overview}

Our proof begins with the well-known observation that to enumerate independent sets in a bipartite graph $G=(X\cup Y, E)$ it suffices to enumerate deviations from the ``ideal'' independent set $X$. 
That is, we have the identity
\begin{equation}\label{eq:iGbasic}
    i(G) = \sum_{A\subset X} 2^{|Y\setminus N(A)|},
\end{equation}
because for a fixed $A\subset X$, any vertex of $Y\setminus N(A)$ can be added to $A$ without spanning an edge.
An important achievement of~\cite{JKP19b} is to give a rigorous proof that in bipartite graphs with strong expansion, typical independent sets are small deviations from either $X$ or $Y$. An algorithm follows provided one figures out how to efficiently enumerate the small deviations and quantify their contributions to $i(G)$; this is done in~\cite{JKP19b} by brute force enumeration and cluster expansion. 
Intuitively, we see a hint of the main idea in equation~\eqref{eq:iGbasic} as when $G$ is an expander we expect that $N(A) \gg |A|$ and so the terms on the right-hand side are small unless $|A|$ is small. 

If the bipartite graph is not an expander, then large deviations from $X$ and $Y$ must be handled. For example, in a $2n$-vertex disjoint union of complete $d$-regular bipartite graphs, a significant number of independent sets intersect both $X$ and $Y$ on $\Omega(n)$ vertices. 
To extend the algorithm to all bipartite graphs, using an idea from~\cite{JPP22} we can separate contributions from expanding and non-expanding pieces of the deviation $A$.
The first step is to break $A\subset X$ in the sum in~\eqref{eq:iGbasic} into pieces with disjoint neighborhoods. 
We say that $A$ is \emph{2-linked} if it is connected in the square $G^2$ of $G$, and note that any $A\subset X$ admits a unique partition into 2-linked subsets that we call \emph{components}. Let $\mathcal{K}(A)$ be the set of components of $A$.
A second refinement step groups the deviations of 2-linked components of $A$ according to their neighborhoods. The \emph{closure} $[A]$ of a set $A\subset X$ is $[A] := \{x\in X : N(x)\subset A\}$, and we say that $A$ is \emph{closed} if $A=[A]$. 
Note that $A$ is closed if and only if each component of $A$ is closed.
Then we have 
\begin{equation}\label{eq:iGclosed}
    i(G) = \sum_{\substack{A\subset X\text{ s.t.\ each}\\\text{component of $A$ is closed}}} \mathcal{D}_A\cdot 2^{|Y\setminus N(A)|},
\end{equation}
where 
\[ \mathcal{D}_A := \prod_{A'\in \mathcal{K(A)}}\bigl| \{ B\subset A' : B~\text{is 2-linked and}~N(B)=N(A')\}\bigr| .\]

A subset $A\subset X$ is called \emph{$t$-expanding} if $|N(A)| = |[A]| + t$, and (in a slight abuse of terminology that we hope the reader permits) $t$-contracting if $|N(A)| < |[A]| + t$. For a fixed $t_0$ that we determine later, we split the sum over $A$ according to $t_0$-contraction.
Then
\begin{equation}\label{eq:iGexpanding}
    i(G) = \sum_{\substack{A\subset X\text{ s.t.\ each}\\\text{component of $A$ is closed} \\ \text{and $t_0$-contracting}}} \mathcal{D}_A\cdot 2^{|Y\setminus N(A)|} \cdot \Xi_A,
\end{equation}
where $\Xi_A$ is defined as follows. For a closed subset $A\subset X$, 
let $X_A = X\setminus N^2(A)$ and $Y_A = Y\setminus N(A)$. 
A \emph{polymer} is a 2-linked subset of $X$, and a tuple of polymers is \emph{compatible} if their neighborhoods are pairwise disjoint. 
Let $\mathcal{P}_A$ be the set of polymers which are subsets of $X_A$.
Then 
\[ \Xi_A := \sum_{k\ge 0}\sum_{\substack{\{B_1,\dotsc,B_k\}\in\mathcal{P}_A\text{ compatible}\\\text{s.t.\ each $B_i$ not $t_0$-contracting}}} 2^{-\sum_{i=1}^k |N(B_i)|},\]
where the inner sum is over unordered tuples of compatible polymers, each of which is not $t_0$-contracting (equivalently, $t$-expanding for some $t\ge t_0$).
For convenience, we define $\mathcal{A}$ to be the set of all $A\subset X$ with closed, $t_0$-contracting components. 
Then the starting point for the analysis of our algorithm is the identity 
\begin{equation}\label{eq:iGidentity}
    i(G) = \sum_{A\in\mathcal{A}} \mathcal{D}_A\cdot 2^{|Y\setminus N(A)|} \cdot \Xi_A,
\end{equation}
derived as above.

Our algorithm simply enumerates the sets $A\in \mathcal{A}$, approximates each $\mathcal{D}_A$ term, and uses the fact (which we must prove) that $1$ is a good approximation of each $\Xi_A$ to approximate $i(G)$. 
The analysis of our algorithm thus splits into three separate components. 
Recall that the input is a $d$-regular bipartite graph $G$ on $2n$ vertices such that for some constant $\delta>0$ we have $d=\lfloor\delta n\rfloor$, and an approximation error $\eps$. We set $t_0=C\log(n/\eps)$, where $C=C(\delta)$ is large enough, and the correctness and running time of our algorithm follows from the results below. 
Note that for this choice of $t_0$ an exponential such as $4^{t_0}$ is polynomial in $n$ and $1/\eps$.

\begin{lemma}\label{lem:enumeratecontracting}
    For $t_0 \leq 2^{-8}d$, the set $\mathcal{A}=\{A\subset X : A\text{ closed and $t_0$-contracting}\}$ has size at most $n^{O(1/\delta)}\cdot 4^{t_0}$ and can be enumerated in the same time.
\end{lemma}

The proof of this lemma uses subspace enumeration to find small cuts in $G$, and then for each such small cut enumerates the sets $A\in\mathcal{A}$ which are close to the cut. See Section~\ref{sec:cuts}.

\begin{restatable}{lemma}{lemapproximateDA}\label{lem:approximateDA}
    Let $A\subset X$ be a 2-linked, closed $t_0$-contracting set. Then for $\eps', \rho>0$ there is a randomized algorithm running in time polynomial in $n$, $1/\eps'$ and $\log(1/\rho)$ that with probability at least $1-\rho$ outputs a relative $\eps'$-approximation to the number of 2-linked subsets $B$ of $A$ such that $N(B)=N(A)$.
\end{restatable}

This lemma uses straightforward estimation of an expectation by repeated sampling, and is very similar to the analogous result in~\cite{JPP22}.
Observe that if $A\in\mathcal A$ has $\ell$ components then running this algorithm on each component with $\eps'\le \eps/(2\ell)$ yields a relative $\eps/2$-approximation to $\mathcal{D}_A$. We use the upper bound $\ell \le 2/\delta = O(1)$ which holds because any $t_0$-contracting set must have size at least $d-t_0\ge d/2$ (this inequality requires that $\eps$ is not exponentially small, but in this case, we can solve the problem exactly by brute force instead).
The proof is in Section~\ref{sec:approxDA}.

\begin{lemma}\label{lem:polymers}
    Let $A\in\mathcal{A}$, then $1\le \Xi_A\le e^{\eps/2}$.
\end{lemma}

This result means that $1$ is a relative $\eps/2$-approximation for each of the $\Xi_A$ terms appearing in~\eqref{eq:iGidentity}. 
The proof is based on graph container methods due to Sapozhenko~\cite{S87,S01b}, which have since been refined,~\cite{GT06},~\cite{GALVIN08}~\cite{KP22},~\cite{PARK22}, and their application to algorithmic counting~\cite{JP20},~\cite{JPP21},~\cite{JPP22}. We give the proof in Section~\ref{sec:containers}.

\section{The algorithm and proof of Theorem~\texorpdfstring{\ref{thm:main}}{\ref*{thm:main}}}

\begin{description}
\item[Input] A $\lfloor\delta n\rfloor$-regular bipartite graph $G=(X\cup Y, E)$ on $2n$ vertices and an approximation error $\eps>0$.
\item[Output] A relative $\eps$-approximation $i'$ of $i(G)$.
\end{description}

Recall that $C=C(\delta)$ is a large enough constant, and that $t_0=C\log(n/\eps)$. In the following proof, implicit constants in the $O(\cdot)$ notation are allowed to depend on $\delta$ but not $\eps$.
If $\eps \le n\exp(-d /(2^8C))$ then we can afford to run a brute force algorithm that computes $i(G)$ exactly in time $e^{O(n)}$ and the running time is still polynomial in $1/\eps$. 
Otherwise, we note that for all large enough $n$ we have $d-2^7 t_0\ge d/2$ and run the following algorithm.
For convenience, we assume that $\eps\le 1$ and simply run the algorithm for $\eps=1$ if the given $\eps$ is larger.

First, construct the set $\mathcal A$, which can be done in time $(n/\eps)^{O(1)}$ by Lemma~\ref{lem:enumeratecontracting}. 
Note also that $|\mathcal{A}|$ is polynomial in $n$ and $1/\eps$.
Then, for each $A\in\mathcal A$ compute an approximation $\tilde{\mathcal D}_A$ of $\mathcal D_A$ by running the algorithm of Lemma~\ref{lem:approximateDA} for each component of $A$ with $\eps'=\eps\delta/4=\Theta(\eps)$ and $\rho=(n/\eps)^{-C'}$ for a large enough constant $C'$.
Recall that there are at most $O(1)$ components of each such $A$, and note that the total number of times the algorithm of Lemma~\ref{lem:approximateDA} is used is thus $(n/\eps)^{O(1)}$. 
By a union bound, with probability at least $3/4$ we get the desired approximation in each application of the lemma, and thus a valid relative $\eps/2$-approximation $\tilde{\mathcal D}_A$ of each $\mathcal D_A$.
Then output $i'=\sum_{A\in\mathcal{A}}\tilde{\mathcal D}_A 2^{|Y\setminus N(A)|}$.
By Lemma~\ref{lem:polymers} and the analysis above the output is a valid $\eps$-approximation of $i(G)$ obtained in time $(n/\eps)^{O(1)}$, thus proving Theorem~\ref{thm:main}.

\section{Subspace enumeration and contracting sets}\label{sec:cuts}

The proof of Lemma~\ref{lem:enumeratecontracting} has two parts. First, we show how to enumerate small cuts using subspace enumeration. For related results see~\cite{ABS10,Kol10,KT07}. 
We use the term \emph{cut} to mean a subset of $V=X\cup Y$, and the \emph{value} $|\nabla(C)|$ of a cut $C$ is the number of edges with precisely one endpoint in $C$.

\begin{lemma}\label{lem:enumeratecuts}
Let $G=(V,E)$ be a $d$-regular bipartite graph on $N=2n$ vertices. 
There is a set $\Ccut\subset 2^V$ such that $|\Ccut|\le n^{O(1/\delta)}$ and $\Ccut$ has the following property.
For all $t\ge 1$ and cuts $S\subset V$ with value $|\nabla(S)|\le td$, there is some $C \in \Ccut$ such that $|S \triangle C| \leq 32t$ and $|\nabla(C)|\le 33td$.
Moreover, the set $\Ccut$ can be constructed in time $n^{O(1/\delta)}$ and hence $|\Ccut|\le n^{O(1/\delta)}$.
\end{lemma}

\begin{proof}
Let $d = \lam_1 \ge \dotsb \ge \lam_{N} = -d$ be the spectrum of the adjacency matrix $A$ of $G$. The facts that $\lam_1=d=\lam_{N}$ and that the spectrum of $A$ is symmetric about zero are standard, see e.g.~\cite{C96}.
Let $k$ be such that $A$ has precisely $2k$ eigenvalue of absolute value at least $d/2$.
Counting closed walks of length two gives
\[ \Tr(A^2) = Nd = \sum_{i=1}^N \lam_i^2 \ge k d^2/2,\]
and hence $k \le 4n/d=O(1/\delta)$.

Let $L=dI-A$ be the Laplacian matrix of $G$ and let $\mathbf e_1,\dotsc,\mathbf e_{N}$ be an orthonormal basis of eigenvectors of $L$ such that $\mathbf e_i$ has eigenvalue $\mu_i$ with $0 = \mu_1 \le \dotsb \le \mu_N = 2d$. 
By the definition of $k$, it must be the case that $\mu_{k+1} > d/2$.
Let $U$ be the span of $\mathbf e_1,\dotsc \mathbf e_k$, and $U^\perp$ be the orthogonal complement of $U$.
For $\eps=\sqrt{2}$, we require an efficient construction of an $\eps$-net $\mathcal E\subset U$ covering all vectors of $L^2$-norm at most $\sqrt n$ in $U$ for some $\eps$. 
For example, we can take 
\[ \mathcal E:= \left\{ \mathbf p = \sum_{i=1}^k x_i\mathbf e_i : x_1,\dotsc,x_k\in \eps/\sqrt{k}\cdot \mathbb{Z},\, \|\mathbf p\|\le \sqrt n \right\},\]
yielding $|\mathcal E|\le (2\sqrt{nk}/\eps)^k$. 
Then every vector in $U$ with $L^2$-norm at most $\sqrt n$ lies at most distance $\eps$ from a vector in $\mathcal E$.

The algorithm to construct $\Ccut$ is as follows. Start with $\Ccut = \emptyset$ and 
for each point $\mathbf{p}\in \mathcal E$, form $\mathbf{p}'$ by rounding each coordinate of $\mathbf{p}$ to $\{0,1\}$ (breaking ties with $1/2 \mapsto 1$) and add the vertex subset with indicator vector $\mathbf p'$ to $\Ccut$.

We now show that $\Ccut$ has the desired properties. 
By the construction of $\Ccut$ and $\mathcal E$ we have $|\Ccut| \le |\mathcal E| \le n^{O(1/\delta)}$.
To establish the other property of $\Ccut$, let $t\ge 1$ and consider an arbitrary subset $S\subset V$ with $|\nabla(S)|\le td$.
Let $\mathbf s$ be the indicator vector of the set $S$ and write this vector in the eigenbasis of $L$ as $\mathbf s =\sum_{i=1}^N s_i \mathbf e_i$. 
Let $\mathbf u = \sum_{i=1}^k s_i \mathbf e_i$ be the projection of $\mathbf s$ onto $U$ and let $\mathbf p$ be the point in $\mathcal E$ closest to $\mathbf u$.
Indicator vectors of subsets of $V$ have $L^2$-norm at most $\sqrt n$, and hence $\|\mathbf u-\mathbf p\|\le \eps$.

Without considering our need for an efficient construction, the idea is that because $\nabla(S)$ is small we know that $\mathbf s$ is an indicator vector close to its projection $\mathbf u$ onto $U$. 
Thus, if we form $\Ccut$ as the union of all sets whose indicator vectors are close to vectors in $U$, each set $S$ of interest has an indicator vector that lies within a distance twice the definition of ``close'' to a set in $C$.

To make the above sketch efficient, we replace $U$ with the $\eps$-net $\mathcal P$. Note that 
\begin{align*}
    td \ge |\nabla(S)| = \mathbf s^{\transpose} L \mathbf s
    =\sum_{i=1}^N \mu_i s_i^2 \ge \frac{d}{2}\sum_{i=k+1}^N s_i^2.
\end{align*}
But $\sum_{i=k+1}^N s_i^2 = \|\mathbf s-\mathbf u\|^2$, so we have the bound $\|\mathbf s-\mathbf u\|\le \sqrt{2t}$.
Then we immediately have $\|\mathbf s - \mathbf p\| \le \sqrt{2t}+\eps$ from the triangle inequality.
Let $\mathbf p'$ be obtained from $\mathbf p$ by rounding each coordinate to $\{0,1\}$, breaking ties with $1/2 \mapsto 1$, and let $C\subset V$ be the set whose indicator vector is $\mathbf p'$. 
We have $|S\triangle C| = \|\mathbf s-\mathbf p'\|^2$ and we bound the latter with the triangle inequality.
In particular, $\mathbf s$ is an indicator vector of distance at most $\sqrt{2t}+\eps$ from $\mathbf p$ and $\mathbf p'$ must be the closest indicator vector to $\mathbf p$, hence $\|\mathbf p-\mathbf p'\|\le \sqrt{2t}+\eps$. 
Then $\|\mathbf s-\mathbf p\| \le 2(\sqrt{2t} + \eps)$, and because $t\ge 1$ and $\eps=\sqrt{2}$ we have 
\[ |S\triangle C|\le 4(\sqrt{t/\delta} + \eps)^2 \le 32t. \]
It remains to bound the value of the cut $|\nabla(C)|$, and the desired bound follows from the observation that
\[ |\nabla(C)| \le |\nabla(S)|+d|S\triangle C| \le td + 32td = 33td.\qedhere\]
\end{proof}

Lemma~\ref{lem:enumeratecuts} tells us that there is an efficient construction of a collection $\Ccut$ of cuts such that any small cut $S$ must be close to a cut in $\Ccut$ in Hamming distance.
We now show that given a small cut $S$ we can enumerate the sets $A\in\mathcal{A}$ which are close to $S$. For this to be useful, it must be that each $A\in\mathcal{A}$ is close to some small cut, and we give the details of this later.

\begin{lemma}\label{lem:dense_enumerate_nonexp}
Fix any $c\geq 1$ and let $t \leq \frac{d}{8c}$. Given a cut $C$ with value at most $td$, there are at most $4^t$ closed $t$-contracting subsets $A\subset X$ such that $|A \triangle (C \cap X)|\leq ct$ and $|N(A) \triangle (C \cap Y)| \leq ct$. 
Moreover, these sets $A$ can be enumerated in time $4^t \cdot n^{O(1)}$.
\end{lemma}

\begin{proof}
Let $A' := C \cap X$ and $W' := C \cap Y$. 
By the fact that $G$ is $d$-regular, $|E(A',W')|\le d\min\{|A'|,|W'|\}$ and hence $|\nabla(C)|\ge d \max\bigl\{|W'|-|A'|,|A'|-|W'|\bigr\}$. 
By assumption, we have $|\nabla(C)| \leq td$ and therefore $\bigl||W'| - |A'|\bigr|\leq t$.

Set
\begin{align*}
S_X & := \bigl\{v \in X \setminus A' : |N(v)\setminus W'| \le 3ct\bigr\},~\text{and} \\
S_Y &:= \bigl\{v \in  W' : |N(v)\cap A'| \le ct\bigr\},
\end{align*}
so that $S_X\subset X$ consists of vertices in $X\setminus A'$ with almost all of their neighbors in $W'$ and $S_Y\subset Y$ consists of vertices in $W'$ with almost all of their neighbors in $X\setminus A'$.
We have the following claims.
\begin{claim}\label{claim:dense_inSX} For any closed, $t$-contracting subset $A\subset X$ such that $|A \triangle A'| \leq ct$, $A \setminus A' \subseteq S_X$.
\end{claim}
\begin{claimproof}
Suppose for contradiction that there is a vertex $v \in A \setminus A'$ such that 
\[ |N(v)\setminus W'|>3ct. \]
We derive the contradiction using the facts that $|\nabla(A\cap A')|=d|A\cap A'|$ and that any of the edges in $\nabla(A\cap A')$ not incident to $W'$ contribute to the value of the cut $C$. These facts imply that $|E(A\cap A',W')|\ge d|A\cap A'|-t\cdot d$, and hence
\[ |N(A\cap A')\cap W'|\ge |A\cap A'|-t. \]
Then because $A$ is closed and non-expanding,
\begin{align*}
|A|+t &\ge |N(A)| \ge |N((A\cap A')\cup\{v\})| \\&> |N((A\cap A')\cup\{v\})\cap W'| + 3ct \\&\ge |A\cap A'| + 2ct\ge |A| + ct,
\end{align*}
 which is a contradiction because there is a strict inequality in the chain and $c\geq 1$.
\end{claimproof}

\begin{claim}\label{claim:dense_inSY} For any $t$-contracting subset $A\subset X$ such that $|A \triangle A'| \leq ct$, $W' \setminus  N(A \cap A')\subseteq S_Y$.
\end{claim}
\begin{claimproof}
We note that for each vertex $v$ in $W' \setminus N(A \cap A')$, we have that 
\[
N(v) \cap A' \subseteq A' \setminus A.
\]
Since $|A' \setminus A| \leq ct$, it follows that $|N(v)\cap A'|\leq ct $.
\end{claimproof}

We can now complete the proof of the lemma. Using the degree constraints in the definitions of $S_X$ and $S_Y$, we have 
\begin{align*}
td& \geq |\nabla(C)| \\
& \geq  |S_X|(d - 3ct) + |S_Y|(d - ct)
\\&\ge (d/2)\cdot(|S_X|+|S_Y|)
\end{align*}
where the last inequality uses $t < \frac{d}{8c}$. As a result, we have
\[
|S_X| + |S_Y| \le 2t.
\]

Putting Claim~\ref{claim:dense_inSX} and Claim~\ref{claim:dense_inSY} together, we have that each closed $t$-contracting sets $A$ with $|A \triangle A'|,~|N(A) \triangle W'| \leq ct$ must be of the form
\[
A = [(A'\setminus N(S_Y')) \cup S_X']
\]
for some subsets $S_Y' \subseteq S_Y$ and $S_X' \subseteq S_X$. Thus, the total number of such $A$ is at most $2^{|S_X|+|S_Y|} \le 4^{t}$.
Since we are given the cut $C$, $S_X$ and $S_Y$ can be found in time polynomial in $n$ as required.
\end{proof}

With these ingredients we can proof Lemma~\ref{lem:enumeratecontracting}, which we recall states that $\mathcal{A}$ can be enumerated in time $n^{O(1/\delta)} 4^{t_0}$.

\begin{proof}[Proof of Lemma~\ref{lem:enumeratecontracting}]
Since $d=\lfloor\delta n\rfloor$, we construct $\Ccut$ as in Lemma~\ref{lem:enumeratecuts} in time $n^{O(1/\delta)}$. 
We then choose $c=32$ and enumerate for each $C\in \Ccut$, every closed $t_0$-contracting subset $A$ with $|A \triangle (C \cap X)|\leq 32t_0$ and $|N(A) \triangle (C \cap Y)| \leq 32t_0$ using Lemma~\ref{lem:dense_enumerate_nonexp}. 
We are done if every $A\in\mathcal A$ appears in this enumeration process, as the running times combine to give the required $n^{O(1/\delta)} 4^{t_0}$.
This holds because each $A\in\mathcal A$ is closed and $t_0$-contracting and hence setting $S_A=A\cup N(A)$ we have $\nabla(S) \le t_0d$. This is because $d|A|$ edges lie between $A$ and $N(A)$ and $|N(A)|< t_0$. 
So each $A\in\mathcal A$ corresponds to a cut of value at most $t_0d$ 
and hence some $C\in\Ccut$ has $|S_A\triangle C|\le 32t_0$ by Lemma~\ref{lem:enumeratecuts}. 
\end{proof}

\section{Approximating the number of covers}\label{sec:approxDA}

For convenience, we restate Lemma~\ref{lem:approximateDA} here.
\lemapproximateDA*

\begin{proof}
    The method is exactly the same as~\cite[Lem.~17]{JPP22}, but in our setting with $d=\lfloor\delta n\rfloor$ the resulting algorithm runs in time polynomial in $n$.
    
    Let $|A|=a$, $N(A)=W$ have size $|W|=w$, and let $W'=\{v\in W : |N(v)\cap A|\le d/2\}$ have size $|W'|=w'$. 
    Let 
    \[ \mathcal D = \{ B\subset A : N(B)=W~\text{and $B$ is 2-linked}\} \]
    be the set whose size we wish to estimate. 

    By~\cite[Cor.~10]{JPP22}, there is a 2-linked subset $A'\subset A$ of size at most 
    \[ \frac{2a}{d}\log d + \frac{2w}{d} + 2(w-a) \le \frac{2}{\delta}\left(1 + \log n\right) + 2t_0 \]
    such that $N(A')=W$. 
    Then $|\mathcal D| \ge 2^{a-\left(\frac{2}{\delta}\left(1 + \log n\right)+2t_0\right)}$, because any subset of $A$ which contains $A'$ is 2-linked. 
    Now $|\mathcal D|$ can be estimated to relative error $\eps'$ with probability at least $1-\rho$ by sampling 
    \[ \frac{1}{(\eps')^2}\log(1/\rho) n^{O(1/\delta)}4^{t_0} \]
    subsets of $A$ uniformly at random, and this can be proved with a suitable application of the Chernoff bound.
\end{proof}

\section{Enumerative lemmas}\label{sec:containers}

In this section we prove Lemma~\ref{lem:polymers} which states that for $A\in\mathcal{A}$ we have $1\le \Xi_A \le e^{\eps/2}$.

\begin{proof}[Proof of Lemma~\ref{lem:polymers}]
For the proof, we fix an arbitrary $A\in\mathcal A$. 
The terms in the sum giving $\Xi_A$ are non-negative, and the lower bound comes from the term $k=0$ which contributes $1$.
For the upper bound, we use recent results on graph containers and adapt them to our purposes.

Recall that a polymer is a 2-linked subset $B\subset X$ and that the function $\Xi_A$ involves a sum over tuples of non-$t_0$-contracting polymers. 
For convenience, we define $\mathcal{G}(w,t)$ to be the set of $t$-expanding polymers with neighborhood size $w$,
\[
\mathcal{G}(w,t) = \{B \subseteq X,\text{ polymer}:|N(B)| = w, |N(B)| - |[B]| = t\}.
\]
In terms of this notation, we have
\begin{align}
 \Xi_A &= \sum_{k\ge 0}\sum_{\substack{\{B_1,\dotsc,B_k\}\in\mathcal{P}_A\text{ compatible}\\\text{s.t.\ each $B_i$ not $t_0$-contracting}}} 2^{-\sum_{i=1}^k |N(B_i)|}
 \\&\le \sum_{k\ge 0}\frac{1}{k!}\left(\sum_{t\ge t_0}\sum_{w\ge 0}|\mathcal{G}(w,t)|2^{-w}\right)^k,
\end{align}
where we drop the requirement on the tuples of being compatible and relax the requirement that the $B_i$ are subsets of $X_A$ to being subsets of $X$, and hence have an upper bound. 
To proceed, we require upper bounds on $|\mathcal{G}(w,t)|$ and split into two cases according to $t$. The following results are proved in the rest of this section and Appendix~\ref{app:deferred}.

\begin{lemma}\label{lem:smalltcontainers}
There is an absolute constant $\gamma>0$ such that for $t_0\leq t \leq \log^4 n$, and any integer $w$,
\[
|\mathcal{G}(w,t)| \leq 2^{w - \gamma t}.
\]
\end{lemma}

\begin{restatable}{lemma}{lemlargetcontainers}\label{lem:largetcontainers}
There is an absolute constant $\gamma>0$ such that for $t \geq \log^4 n$, and any integer $w$,
\[
|\mathcal{G}(w,t)| \leq 2^{w - \gamma t}.
\]
\end{restatable}

We prove Lemma~\ref{lem:smalltcontainers} below with some aspects of the container method that are somewhat standard deferred to the appendix. The proof of Lemma~\ref{lem:largetcontainers} is a simple application of a result in~\cite{JPP22} which we give now.

\begin{proof}[Proof of Lemma~\ref{lem:largetcontainers}]
For each $v \in V$, let us define
\[
\mathcal{G}'(v,w,t) = \{A \in \mathcal G(w,t) : v \in A\}.
\]
First, we observe that $\log^2d \cdot \frac{t}{d} \leq \log^2 n\cdot \frac{n}{\delta n} \ll \log^4 n$. Lemma 4 in~\cite{JPP22} gives us that there is a constant $c$ such that for each $v$, $\mathcal{G}'(v,w,t) \leq 2^{w - ct}$. Thus, we have
\[
|\mathcal{G}(w,t)| \leq \sum_{v}|\mathcal{G}'(v,w,t)| \leq n\cdot 2^{n - ct} \leq 2^{n - ct/2}
\]
for $n$ large enough. Setting $\gamma = c/2$ completes the proof.
\end{proof}

With these lemmas in hand, and because each neighborhood size $w$ that we see is in $[1,n]$, there is an absolute constant $\gamma>0$ such that 
\begin{align}
 \Xi_A &\le \sum_{k\ge 0}\frac{1}{k!}\left(\sum_{t\ge t_0}n2^{-\gamma t}\right)^k
 \\&=\sum_{k\ge 0}\frac{1}{k!}\left(n\frac{2^{-\gamma t_0}}{1-2^{-\gamma}}\right)^k = \exp\left(n\frac{2^{-\gamma t_0}}{1-2^{-\gamma}}\right).
\end{align}
This at most the required $e^{\eps/2}$ provided that 
\[ t_0 \ge \frac{1}{\gamma}\log_2\left(\frac{2}{1-2^{-\gamma}}\frac{n}{\eps}\right), \]
which our choice $t_0=C\log(n/\eps)$ satisfies for all large enough constants $C=C(\delta)$.
\end{proof}

The rest of this section is dedicated to the proofs of Lemmas~\ref{lem:smalltcontainers} and~\ref{lem:largetcontainers}.
Note that $\mathcal{G}(w,t)$ is a collection of subsets of $X$ and we are no longer fixing some $A\in\mathcal A$ and focusing on subsets of $X_A$. 

Given a vertex $v\in V$ and a subset $S\subset V$, we write $d_S(v)$ for the number of neighbors of $v$ in $S$. 
For a subset $A\subset X$, we write $W=N(A)$ and $W_s=\bigl\{y\in W : d_A(y)\ge s\bigr\}$. 
We say that $F$ is an \emph{essential set} for $A$ if $W \supseteq F \supseteq W_{d/2}$ and $N(F)\supseteq [A]$. 
It may be useful to consider such an $F$ an approximation for the neighborhood $W=N(A)$. 
We call a tuple $(S,T) \in 2^X \times 2^Y$ a $\gamma'$-container for a $t$-contracting subset $A\subset X$ with neighborhood $W=N(A)$ if
\begin{enumerate}
\item\label{itm:containermembership} $S \supseteq [A]$ and $W_{d/2} \subseteq T \subseteq W$,
\item\label{itm:containerSdegrees}  $d_{Y \setminus T}(v)\leq \gamma' t$ for each $v \in S$, and
\item\label{itm:containerYTdegrees}  $d_{S}(v) \leq \gamma' t$ for each $v \in Y \setminus T$.
\end{enumerate}
The following two results show the existence of containers and bound the number of sets for which a given container is a $\gamma'$-container.

\begin{restatable}{lemma}{lemcontainer}\label{lem:container}
For any $\gamma' > 0$ and any set $F \subseteq Y$, there is a set $\Cind \subseteq 2^X \times 2^Y$ of size at most $n^{O(1/\gamma')}$ such that any $A \subseteq X$ for which $F$ is an essential set, has a $\gamma'$-container in $\Cind$.
\end{restatable}

\begin{restatable}{lemma}{lemreconstruction}\label{lem:reconstruction}
There is an absolute constant $\gamma''>0$ such that the following holds:

For any $\gamma'>0$ any $w<n$, $t < \log^4 n$ and tuple $(S,T)\in \Cind$, there are at most $2^{w - \gamma'' t}$ sets $A \in \mathcal{G}(w,t)$ such that $(S,T)$, and is a $\gamma'$-container for $A$.
\end{restatable}

Since the proofs of these results are small modifications of existing container results, e.g.~\cite{PARK22}, we defer their proofs to Appendix~\ref{app:deferred}. 
We are now ready to handle the case of small $t$ and prove Lemma~\ref{lem:smalltcontainers}. We leave the proof of Lemma~\ref{lem:largetcontainers} to Appendix~\ref{app:deferred}.

\begin{proof}[Proof of Lemma~\ref{lem:smalltcontainers}]
Consider an integer $t \in [t_0,\log^4 n]$ and a set $A\in \mathcal{G}(w,t)$. Define $L := [A] \cup N(A)$. The cut $L$ has value at most $td$. 
By Lemma~\ref{lem:enumeratecuts}, there is a cut $L'\in\Ccut$ such that $g := |L \triangle L'| \leq O(t)$. Let $A' := L' \cap X$ and $W' := L' \cap Y$. 

Consider the set $W'_g = \{u \in Y : d_{A'}(u)> g\}$. We have the following two claims.

\begin{claim}\label{claim:essential1}
$W \supseteq W'_g \supseteq \{u \in W : d_{A'}(u) \geq d/2\}$.
\end{claim}

\begin{claimproof}
Consider a vertex $u \in W$ such that $d_A(u)\geq d/2$. We have
\[
d_{A'}(u)  \geq d_A(u) -  |A\setminus A'| \geq d/2 - |L \triangle L'| \geq d/2 - g > g,
\]
and hence $u \in W'_g$. Moreover, consider a vertex $u \in W'_g$. We have
\[
d_A(u) \geq d_{A'}(u) - |A' \setminus A| > g - |L \triangle L'| > 0,
\]
and hence $u \in W$.
\end{claimproof}

\begin{claim}\label{claim:essential2}
$A\subseteq N(W'_g)$.
\end{claim}

\begin{claimproof}
Suppose otherwise, i.e.\ there is a vertex $u\in A$ such that for each vertex $v \in N(u)$ we have $d_{A'}(u) \leq g$. For any such $v$, we have 
\[
d_A(v) \leq d_{A'}(v) + |A \setminus A'| \leq d_{A'}(v) + |L \triangle L'| \leq 2g.
\]
This gives us that 
\[
t \cdot d = |E(W,X \setminus A)| \geq |E(N(u), W\setminus A)| \geq d(d - 2g),
\]
contradicting the assumptions that $d=\lfloor\delta n\rfloor$ and $t$ and $g$ are both $O(\log^4 n)$.
\end{claimproof}

Claims~\ref{claim:essential1} and~\ref{claim:essential2} show that $W'_g$ is an essential set for $A$. 
The set $A\in\mathcal{G}(w,t)$ may be constructed by
\begin{enumerate}
\item choosing the appropriate cut $L'$ in the set $\Ccut$ constructed in Lemma~\ref{lem:enumeratecuts},
\item constructing the essential subset $W'_g$ for it as above,
\item using Lemma~\ref{lem:container} to obtain a $\gamma'$-container of $A$, where $\gamma'$ is the absolute constant of Lemma~\ref{lem:reconstruction}, and finally
\item reconstructing $A$ from the $\gamma'$-container with Lemma~\ref{lem:reconstruction}.
\end{enumerate}
There are $n^{O(1/\delta)}$ choices for $L'$ in the first step, a unique construction of $W'_g$ for the second, $n^{O(1/\gamma')}$ possible containers in the third step, and $2^{w - \gamma'' t}$ ways for the final step. In total there are 
\[
2^{w - \gamma'' t +O(1/\gamma' + 1/\delta)\log n} \le 2^{w - \gamma''t/2}
\]
such sets $A\in\mathcal{G}(w,t)$.
The last inequality comes from our assumption that $t \geq t_0$ for our choice of $t_0=C(\delta)\log(n/\eps) \ge C\log(n)$ (because wlog $\eps\le 1$) satisfying 
\[ t_0 \geq \Omega\left(\frac{\log n}{\gamma''}\left(\frac{1}{\gamma'} + \frac{1}{\delta}\right)\right).\qedhere \]
\end{proof}

\section{Concluding remarks and future directions}
\begin{enumerate}[wide, labelwidth=!, labelindent=0pt]
\item[1.] Naturally, a next goal is to understand the power and limitations of the methods presented, especially in conjunction with existing cluster expansion methods. More specifically, we are curious about the following two questions: 
\begin{itemize}
\item[i.] Can this spectral point of view help with our understanding of independent sets in a larger class of bipartite graphs?
\item[ii.] To what extent do these methods help in reducing the computation needed to implement algorithmic cluster expansion? 
\end{itemize}
In this context, the problem of approximating the number of independent sets in \emph{small-set expanders} feels within striking distance.\\
\item[2.] Our next remark concerns Lemma~\ref{lem:enumeratecuts}. As mentioned before, similar results have had other applications in optimization and Unique Games~\cite{ABS10,Kol10,KT07}, though we take a subtly different viewpoint worth noting: we seek to approximate \emph{all} cuts in the graph, not just small ones. In any case, we find the lemma interesting in its own right and conjecture something stronger.
\begin{conjecture} Lemma~\ref{lem:enumeratecuts} holds with $|\Cind| \leq 2^{O(1/\delta)}$.
\end{conjecture}
If true, this would be best possible, as evidenced by a disjoint union of $1/\delta$ components. Setting $t = 0$ in this case gives exactly $2^{1/\delta}$ cuts of size $0$.\\
\item[3.] Finally, we leave open the problem of making our algorithm deterministic. At the moment, the only step where randomness is used is Lemma~\ref{lem:approximateDA}. 
\end{enumerate}

\bibliographystyle{habbrv}
\bibliography{main}

\appendix
\section{Deferred proofs}\label{app:deferred}

\subsection{Proof of Lemma~\texorpdfstring{\ref{lem:container}}{\ref*{lem:container}}}

We restate the result for convenience. 

\lemcontainer*
\begin{proof}
Let $A\subset X$ be a subset for which $F$ is an essential set and let $W=N(A)$, $t := |N(A)| - |[A]|$. Consider the following algorithm
\begin{itemize}
\item[] initialize $T \gets F$
\item[] while $\exists~v \in [A]$ s.t.\ $d_{W \setminus T}(v) > \gamma' t$, pick such a $v$:
\begin{itemize}
\item[] $T \gets T \cup N(v)$
\end{itemize}
\item[] initialize $S \gets \{v \in X : d_{Y \setminus T}(v) \leq \gamma' t\}$
\item[] while $\exists~v \in Y\setminus W$ s.t. $d_{S}(v) > \gamma' t$, pick such a $v$:
\begin{itemize}
\item[] $S \gets S \setminus N(v)$
\end{itemize}
\item[] $T \gets T\cup \{v \in Y : d_S(v) > \gamma' t\}$
\item[] return $(S,T)$
\end{itemize}

The lemma follows provided we can show that $(S,T)$ as given by the algorithm above is a $\gamma'$-container for $A$ by establishing properties~\ref{itm:containermembership}--\ref{itm:containerYTdegrees}, and provided we can show a good enough bound on the total number of outputs $(S,T)$ which can occur for a fixed $F$ as $A$ varies.

To prove that the output $(S,T)$ is a $\gamma'$-container of $A$, we first show that $S \supseteq [A]$ and $W_{d/2} \subseteq T \subseteq W$, establishing~\ref{itm:containermembership}. 
Since $F$ is an essential subset for $A$, we initialize $T\gets F$, and $T$ can then only grow, we have $W_{d/2}\subset T$. Clearly, $T \subseteq W$ at the end of the first while loop.
After the second initialize statement, we have that each vertex $v \in [A]$ satisfies $d_{W \setminus T}(v) \leq d_{Y \setminus T}(v) \leq \gamma' t$. Therefore, $S \supseteq A$ at the end of this line. This property is maintained during the second while loop since we only delete $N(v)$ from $S$ for $v \not\in W$. This also means that in the penultimate line, all vertices added to $T$ are from $W$. Thus $T \subseteq W$ is also maintained at the end of the algorithm.
Next, we prove~\ref{itm:containerYTdegrees}. 
At the beginning of the second loop, every $v \in S$ satisfies $d_{Y \setminus T}(v)\leq \gamma' t$. Since vertices are only removed from $S$ and added to $T$ after this point, this property is preserved till the end. Finally, to prove~\ref{itm:containerSdegrees} note that the penultimate line of the algorithm ensures that every $v \in Y \setminus T$ satisfies $d_S(v) \leq \gamma' t$. 

To bound the number of possible outputs for a fixed $F$, note that before the start of the first loop we have $|W \setminus T| \leq O(t)$. Each step in the first loop of the algorithm removes $\gamma t$ vertices from $W \setminus T$. 
Therefore, this loop runs at most $O(1/\gamma')$ times. Next, each step in the second loop removes at least $\gamma' t$ vertices from $S \setminus [A]$. Immediately after the second initialize statement, we have
\[
dt\geq |E(S \setminus [A], T)| \geq (d - \gamma' t)|S \setminus [A]|.
\]
As a result, $|S \setminus [A]| = O(t)$. So the second loop runs for at most $1/\gamma'$ steps. The output is determined by the set of $O(1/\gamma')$ vertices chosen in both loops, so the number of possible outputs for the algorithm for a given $F$ is at most $n^{O(1/\gamma')}$.
\end{proof}

\subsection{Proof of Lemma~\texorpdfstring{\ref{lem:reconstruction}}{\ref*{lem:reconstruction}}}

We restate the result for convenience. 
\lemreconstruction*

We need the following lemma
\begin{lemma}
Let $(S,T)$ be a $\gamma'$-container for a set $A \in \mathcal{G}(w,t)$. Then $|S| \leq |T|$.
\end{lemma}

\begin{proof}
Let us denote $W = N(A)$.
First, we observe that $|E(S,W)| \leq d|T| + \gamma' t |W \setminus T|$ by~\ref{itm:containerYTdegrees}. We also have that $|E(S,W)| \geq d|[A]| + |S \setminus [A]|(d - \gamma' t) = d|S| - \gamma' t |S \setminus [A]|$ by~\ref{itm:containermembership} and~\ref{itm:containerSdegrees}. Combining these inequalities, we have
\begin{equation}\label{eqn:containersizes}
|S| \leq |T| + \frac{\gamma' t(|S \setminus [A]| + |W \setminus |T||)}{d}.
\end{equation}
Since $T \supseteq W_{d/2}$, we have that $|W \setminus T| \leq O(t)$ and
\begin{align*}
td & = |E(W,X\setminus [A])| \geq \sum_{v \in S \setminus [A]}d_{T}(v) \geq |S \setminus [A]|(d - \gamma' t)
\end{align*}
which gives $|S \setminus [A]| = O(t)$. So~\eqref{eqn:containersizes} implies
\[
|S| \leq |T| + O\left(\frac{\gamma' t^2}{d}\right).
\]
Since $t \leq \log^4 n$, $d=\lfloor\delta n\rfloor$, and $|S|$ and $|T|$ are both integers, we have that $|S| \leq |T|$.
\end{proof}

We finish the proof using the following lemma from~\cite{PARK22}, whose proof we reproduce for clarity.

\begin{lemma}[\cite{PARK22}, Lemma 11]\label{lem:PARK}
There is an absolute constant $\gamma''>0$ such that the following holds:

For any tuple $(S,T)\in 2^X \times 2^Y$ such that $|S| \leq |T|$, there are at most $2^{w - \gamma'' t}$ sets $A \in \mathcal{G}(w,t)$ such that $[A] \subseteq S$ and $T \subseteq N(A)$.
\end{lemma}

To be precise, in~\cite{PARK22} the graph in question is the $d$-dimensional hypercube and additional hypotheses are stated, namely $w-t < n/4$ and $w > d^4$. These play no role in the proof, however, and it extends verbatim to the result stated above. 

\begin{proof}
Throughout, we denote $W = N(A)$, and let $\alpha > 0$ be a constant that will be determined later. 

If $|S| < w - \alpha t$, then $A$ is among the possible $2^{w - \alpha t}$ subsets of $S$. Suppose otherwise, that $|S| > w - \alpha t$. Let $A^* \in \mathcal{G}(w,t)$ such that $(S,T)$ is a $\gamma'$-container for $A^*$ and let $W^* = N(A^*)$. We have that $[A]$ is completely determined by $W \setminus W^*$ and $W^*\setminus W$. Since $W^*\setminus W \subseteq W^*\setminus T$, and 
\[
|W^* \setminus T| \leq |W^*| - |T| = |W| - |T| \leq |W| - |S|  \leq \alpha t,
\]
there are at most $2^{\alpha t}$ choices for $W^* \setminus W$. Next, for each vertex in $W \setminus W^*$, we choose a neighbor in $A \setminus A^* \subseteq S \setminus A^*$. Observe that $W \setminus W^* = N(A \setminus A^*) \setminus W^*$. Since 
\[
|W \setminus W^*| \leq |W \setminus F| = |W| - |F| \leq |W| - |S| \leq \alpha t,
\]
and 
\[
|S \setminus A^*| \leq |S| - |A^*| = |S| - |A| \leq |T| - |A| \leq |W| - |A| =  t.
\]
Therefore, the number of choices for $W \setminus W^*$ is at most
\[
\binom{t}{\alpha t} \leq 2^{H(\alpha) t}.
\]
Once we have $[A]$, there are at most $2^{w - t}$ possibilities for $A$. Thus the total number of choices is at most
\[
2^{w - t + t(\alpha + H(\alpha))}.
\]
Choosing e.g., $\alpha = 0.17$ allows one to choose $\gamma'' = 0.17$.
\end{proof}

\end{document}